\documentclass{llncs} 
\usepackage{llncsdoc} 
\usepackage{enumerate}
\usepackage{graphicx} 
\usepackage{graphics, color}
\usepackage{amsmath} 
\usepackage{algorithm} 
\usepackage{algorithmic}
\usepackage{boxedminipage} 
\usepackage{amssymb}
\usepackage{verbatim}
\usepackage{hyperref} 
\usepackage{array}
\usepackage{multirow}
\usepackage{color}
\usepackage{changebar}
\usepackage{xcolor}

\newcommand{\added}[1]{#1}

\newcommand{\BGP}{\ensuremath{\mathsf{BGP }}}
\newcommand{\XORSample}{\ensuremath{\mathsf{XORSample }}}
\newcommand{\XORSampleprime}{\ensuremath{\mathsf{XORSample}'}}
\newcommand{\BoundedSAT}{\ensuremath{\mathsf{BoundedSAT}}}

\newcommand{\UniformWitness}{\ensuremath{\mathsf{UniWit}}}

\newcommand{\SATModelCount}{\ensuremath{\mathsf{SATModelCount}}}

\newcommand{\mc}[1]{\ensuremath{\mathcal{#1}}}
\newcommand{\prob}{\ensuremath{\mathsf{Pr}}}
\newcommand{\expect}{\ensuremath{\mathsf{E}}}

\newcommand{\killthis}[1]{}

\begin{document}
		
\title{A Scalable and Nearly Uniform Generator of SAT Witnesses\thanks{The final version will appear in the Proceedings of CAV'13 and will be available at \url{link.springer.com}. Work supported in part by NSF grants CNS 1049862 and CCF-1139011,
by NSF Expeditions in Computing project "ExCAPE: Expeditions in Computer
Augmented Program Engineering", by BSF grant 9800096, by gift from
Intel, by a grant from Board of Research in Nuclear Sciences, India, and by the Shared University Grid at Rice funded by NSF under Grant EIA-0216467, and a partnership between Rice University, Sun Microsystems, and Sigma Solutions, Inc.}}
\author{Supratik Chakraborty\inst{1} \and Kuldeep S. Meel\inst{2} \and Moshe Y. Vardi\inst{2}}
\institute{Indian Institute of Technology Bombay, India
\and
Department of Computer Science, Rice University}
\date{Sept.~6, 2012}

\maketitle
\pagestyle{empty}
\begin{abstract}
Functional verification constitutes one of the most challenging tasks
in the development of modern hardware systems, and simulation-based
verification techniques dominate the functional verification
landscape. A dominant paradigm in simulation-based verification is directed
random testing, where a model of the system is simulated with a set of 
random test stimuli that are uniformly or near-uniformly distributed 
over the space of all stimuli satisfying a given set of constraints. 
Uniform or near-uniform generation of solutions for
large constraint sets is therefore a problem of theoretical and
practical interest. For boolean constraints, prior work offered
heuristic approaches with no guarantee of performance, and theoretical
approaches with proven guarantees, but poor performance in
practice. We offer here a new approach with theoretical performance
guarantees and demonstrate its practical utility on large constraint
sets.
\end{abstract}
\section{Introduction}
Functional verification constitutes one of the most challenging tasks
in the development of modern hardware systems.  Despite significant
advances in formal verification over the last few decades, there is a
huge mismatch between the sizes of industrial systems and the
capabilities of state-of-the-art formal-verification
tools~\cite{Ben05}.  Simulation-based verification techniques
therefore dominate the functional-verification
landscape~\cite{chang2008functional}.  A dominant paradigm in
simulation-based verification is directed random testing.  In this
paradigm, an operational (usually, low-level) model of the system is
simulated with a set of random test stimuli satisfying a set of
\emph{constraints}
\cite{chandra-Verification,NRJKVMS06,Yuan2004}.  The
simulated behavior is then compared with the expected behavior, and
any mismatch is flagged as indicative of a bug.  The constraints that
stimuli must satisfy typically arise from various sources such as
domain and application-specific knowledge, architectural and
environmental requirements, specifications of corner-case scenarios,
and the like.  Test requirements from these varied sources are
compiled into a set of constraints and fed to a constraint solver to
obtain test stimuli.  Developing constraint solvers (and test
generators) that can reason about large sets of constraints is
therefore an extremely important activity for industrial test and
verification applications~\cite{GABK11}.

Despite the diligence and insights that go into developing constraint
sets for generating directed random tests, the complexity of modern
hardware systems makes it hard to predict the effectiveness of any
specific test stimulus.  It is therefore common practice to generate a
large number of stimuli satisfying a set of constraints.  Since every
stimulus is \emph{a priori} as likely to expose a bug as any other
stimulus, it is desirable to sample the solution space of the
constraints uniformly or near-uniformly (defined formally below) at
random~\cite{NRJKVMS06}.  A naive way to accomplish this is to first
generate all possible solutions, and then sample them uniformly.
Unfortunately, generating all solutions is computationally prohibitive
(and often infeasible) in practical settings of directed random
testing.  For example, we have encountered systems of constraints
where the expected number of solutions is of the order of $2^{100}$,
and there is no simple way of deriving one solution from another.  It
is therefore interesting to ask: \emph{Given a set of constraints, can
  we sample the solution space uniformly or near-uniformly, while
  scaling to problem sizes typical of testing/verification scenarios?}
An affirmative answer to this question has implications not only for
directed random testing, but also for other applications like
probabilistic reasoning, approximate model counting and Markov logic
networks~\cite{Bacchus2003,Roth1996}.

In this paper, we consider Boolean constraints in conjunctive normal
form (CNF), and address the problem of near-uniform generation of
their solutions, henceforth called \emph{SAT Witnesses}. \added{This
  problem has been of \added{long-standing theoretical}
  interest~\cite{Sipser83,Stockmeyer83}.  Industrial approaches to
  solving this problem either rely on ROBDD-based
  techniques~\cite{Yuan2004} \added{, which do not scale} well (see, for example,
  the comparison in~\cite{KitKue2007}), or use heuristics that offer
  no guarantee of performance or uniformity when applied to large
  problem instances\footnote{Private communication: R. Kurshan}.}
Prior published work in this area broadly belong to one of two
categories.  In the first category
\added{~\cite{Selman-Sampling,Kirkpatrick83,Gomes-Sampling,KitKue2007}},
the focus is on heuristic sampling techniques that scale to large
systems of constraints.  Monte Carlo Markov Chain (MCMC) methods and
techniques based on random seedings of SAT solvers belong to this
category.  However, these methods either offer very weak or no
guarantees on the uniformity of sampling (see~\cite{KitKue2007} for a
comparison), or require the user to provide hard-to-estimate
problem-specific parameters that crucially affect the performance and
uniformity of sampling.  In the second category of
work~\cite{Bellare98uniformgeneration,Jerr,Yuan2004}, the focus is on
stronger guarantees of uniformity of sampling.  Unfortunately, our
experience indicates that these techniques do not scale even to
relatively small problem instances (involving few tens of variables)
in practice.

The work presented in this paper tries to bridge the above mentioned
extremes.  Specifically, we provide guarantees of near-uniform
sampling, and of a bounded probability of failure, without the user
having to provide any hard-to-estimate parameters.  We also
demonstrate that our algorithm scales in practice to constraints
involving thousands of variables. Note that there is evidence that
uniform generation of SAT witnesses is harder than SAT solving
\cite{Jerr}.  Thus, while today's SAT solvers are able to handle
hundreds of thousands of variables and more, we believe that scaling of
our algorithm to thousands of variables is a major improvement in this
area.  Since a significant body of constraints that arise in
verification settings and in other application areas (like
probabilistic reasoning) can be encoded as Boolean constraints, our
work opens new directions in directed random testing and in these
application areas.

The remainder of the paper is organized as follows.  In
Section~\ref{sec:prelims-for-bellare}, we review preliminaries and
notation needed for the subsequent discussion. In
Section~\ref{sec:bellare-algorithm}, we give an overview of some
algorithms presented in earlier work that come close to our work.
Design choices behind our algorithm, some implementation issues, and a
mathematical analysis of the guarantees provided by our algorithm are
discussed in Section~\ref{sec:our-algorithm}.
Section~\ref{sec:experiments} discusses experimental results on a
large set of benchmarks.  Our experiments demonstrate that our
algorithm is more efficient in practice and generates witnesses that
are more evenly distributed than those generated by the best known
alternative algorithm that scales to comparable problem sizes.
Finally, we conclude in Section~\ref{sec:conclusion}.

\section{Notation and Preliminaries}\label{sec:prelims-for-bellare}

Our algorithm can be viewed as an adaptation of the algorithm proposed
by Bellare, Goldreich and Petrank~\cite{Bellare98uniformgeneration}
for uniform generation of witnesses for \mc{NP}-relations. In the
remainder of the paper, we refer to Bellare et al.'s algorithm as
the {\BGP} algorithm (after the last names of the authors). Our
algorithm also has similarities with algorithms presented by Gomes,
Sabharwal and Selman~\cite{Gomes-Sampling} for near-uniform sampling
of SAT witnesses.  We begin with some notation and preliminaries
needed to understand these related work.

Let $\Sigma$ be an alphabet and $R \subseteq \Sigma^* \times \Sigma^*$
be a binary relation.  We say that $R$ is an \mc{NP}-relation if $R$
is polynomial-time decidable, and if there exists a polynomial
$p(\cdot)$ such that for every $(x, y) \in R$, we have $|y| \le
p(|x|)$.  Let $L_R$ be the language $\{x \in \Sigma^* \mid \exists y
\in \Sigma^*,\, (x, y) \in R\}$.  The language $L_R$ is said to be in
\mc{NP} if $R$ is an \mc{NP}-relation.  The set of all satisfiable
propositional logic formulae in CNF is known to be a language in
\mc{NP}.  Given $x \in L_R$, a \emph{witness} of $x$ is a string $y
\in \Sigma^*$ such that $(x, y) \in R$.  The set of all witnesses of
$x$ is denoted $R_x$.  For notational convenience, let us fix $\Sigma$
to be $\{0, 1\}$ without loss of generality.  If $R$ is an
\mc{NP}-relation, we may further assume that for every $x \in L_R$,
every witness $y \in R_x$ is in $\{0, 1\}^n$, where $n = p(|x|)$ for
some polynomial $p(\cdot)$.
%% Let $\Phi_k$ denote the set of all Boolean CNF formulae over $k$
%% variables, say $v_1, \ldots v_k$, and let $\Pi_k$ denote the set of
%% all assignments of truth values to $(v_1, \ldots v_k)$.  An assignment
%% $\pi \in \Pi_k$ is said to satisfy a formula $\varphi \in \Phi_k$,
%% denoted $\pi \models \varphi$, if $\varphi$ evaluates to
%% $\mathsf{True}$ when $v_1, \ldots v_k$  are assigned values as in
%% $\pi$.  Assuming formulae and assignments are encoded as strings over
%% $\{0, 1\}$, the relation {\sc CNF-SAT} $= \{(\varphi, \pi) \mid
%% \exists k > 0.\, \varphi \in \Phi_k, \pi \in \Pi_k, \mbox{ and } \pi
%% \models \varphi\}$ is known to be an \mc{NP}-relation.

Given an \mc{NP} relation $R$, a \emph{probabilistic generator} of
witnesses for $R$ is a probabilistic algorithm $\mathcal{G}(\cdot)$
that takes as input a string $x \in L_R$ and generates a random
witness of $x$.  Throughout this paper, we use $\prob\left[X\right]$
to denote the probability of outcome $X$ of sampling from a
probability space.  A \emph{uniform generator} $\mathcal{G}^u(\cdot)$
is a probabilistic generator that guarantees
$\prob\left[\mathcal{G}^u(x) = y\right] = 1/|R_x|$ for every witness
$y$ of $x$.  A \emph{near-uniform generator} $\mathcal{G}^{nu}(\cdot)$
relaxes the guarantee of uniformity, and ensures that
$\prob\left[\mathcal{G}^{nu}(x) = y\right] \geq c\cdot(1/|R_x|)$ for a
constant $c$, where $0 < c \leq 1$.  Clearly, the larger $c$ is, the
closer a near-uniform generator is to being a uniform generator.  Note
that near-uniformity, as defined above, is a more relaxed
approximation of uniformity compared to the notion of ``almost
uniformity'' introduced in ~\cite{Bellare98uniformgeneration,Jerr}.
In the present work, we sacrifice the guarantee of uniformity and
settle for a near-uniform generator in order to gain performance
benefits.  Our experiments, however, show that the witnesses generated
by our algorithm are fairly uniform in practice.  Like previous work
~\cite{Bellare98uniformgeneration,Jerr}, we allow our generator to
occasionally ``fail", i.e. the generator may occasionaly output no witness,
but a special failure symbol $\bot$.  A generator that occasionally
fails must have its failure probability bounded above by $d$, where
$d$ is a constant strictly less than $1$.
%We will use $y[i]$ to refer to the $i^{th}$ bit of $y$.
%[***MYVL I don't understand this point ***] 
%This does not sacrifice generality since we can always pad up
%witnesses, if needed.  

A key idea in the {\BGP} algorithm for uniform generation of
witnesses for \mc{NP}-relations is to use $r$-wise independent hash
functions that map strings in $\{0, 1\}^n$ to $\{0, 1\}^m$, for $m \le
n$.  The objective of using these hash functions is to partition $R_x$
with high probability into a set of ``well-balanced'' and ``small''
cells.  We follow a similar idea in our work, although there are
important differences.  Borrowing related notation and terminology
from~\cite{Bellare98uniformgeneration}, we give below a brief overview
of $r$-wise independent hash functions as used in our context.

Let $n, m$ and $r$ be positive integers, and let $H(n,m,r)$ denote a
family of $r$-wise independent hash functions mapping $\{0, 1\}^n$ to
$\{0, 1\}^m$.  We use $h \xleftarrow{R} H(n,m,r)$ to denote the act of
choosing a hash function $h$ uniformly at random from $H(n,m,r)$.  By
virtue of $r$-wise independence, for each $\alpha_1, \ldots \alpha_r
\in \{0,1\}^m $ and for each distinct $y_1, \ldots y_r \in \{0,1\}^n$,
$\prob\left[\bigwedge_{i=1}^r h(y_i) = \alpha_i: h
  \xleftarrow{R} H(n, m, r)\right] = 2^{-mr}$.

For every $\alpha \in \{0, 1\}^m$ and $h \in H(n,m,r)$, let
$h^{-1}(\alpha)$ denote the set $\{y \in \{0, 1\}^n \mid h(y) =
\alpha\}$.  Given $R_x \subseteq \{0, 1\}^n$ and $h \in H(n, m, r)$, we
use $R_{x, h, \alpha}$ to denote the set $R_x \cap h^{-1}(\alpha)$.
If we keep $h$ fixed and let $\alpha$ range over $\{0, 1\}^m$, the
sets $R_{x, h, \alpha}$ form a partition of $R_x$.  Following the
notation of Bellare et al., we call each element of such a parition a
\emph{cell} of $R_x$ induced by $h$.  %% The hash function $h$ is said to
%% induce \emph{small} cells of $R_x$ if $|R_{x,h,\alpha}| \le 2n^2$ for
%% every $\alpha \in \{0, 1\}^m$.  Similarly, it is said to induce
%% \emph{non-trivial} cells of $R_x$ if $|R_{x,h,\alpha}| \ge n^2/2$
%% for every $\alpha \in \{0, 1\}^m$. 
It has been argued
in~\cite{Bellare98uniformgeneration} that if $h$ is chosen uniformaly
at random from $H(n, m,r)$ for $r \ge 1$, the expected size of $R_{x, h,
  \alpha}$, denoted $\expect\left[|R_{x,h,\alpha}|\right]$, is $|R_x|/2^m$, 
for each $\alpha \in \{0, 1\}^m$.
%[MYV: We are going from the probability of a collision to the
%expected size of a cell. This requires some argument. Is there a lemma to
%this effect in Bellare et al.?]

%% We now define when the cells are \emph{small} and \emph{nontrivial}.
%% Our terminology here is a slight variation on that of Bellare et
%% al.~\cite{Bellare98uniformgeneration}.
%% \begin{definition}
%% For $k > 0$ and $0 < \epsilon \le 1$, the cell $R_{x, h, \alpha}$ is said to
%% be $(k, \epsilon)$-small if $|R_{x, h, \alpha}| < (1+\epsilon)\cdot n^{1/k}$, 
%% and it is said to be $(k, \epsilon)$-non-trivial if $|R_{x, h, \alpha}| \ge
%% (1-\epsilon)\cdot n^{1/k}$.  A hash function $h$ is said to induce $(k,
%% \epsilon)$-small (respectively, $(k, \epsilon)$-non-trivial) cells of $R_x$ if
%% $R_{x, h, \alpha}$ is $(k, \epsilon)$-small (respectively, $(k,
%% \epsilon)$-non-trivial) for every $\alpha \in \{0, 1\}^m$.
%% \end{definition}
%Note that the notion of small and non-trivial cells used in Bellare et
%al's work correspond to $(\frac{1}{2}, 1)$-small and $(\frac{1}{2},
%\frac{1}{2})$-non-trivial cells using the above terminology.

In~\cite{Bellare98uniformgeneration}, the authors suggest using
polynomials over finite fields to generate $r$-wise independent hash
functions.  We call these \emph{algebraic} hash functions.  Choosing a
random algebraic hash function $h \in H(n, m, r)$ requires choosing a
sequence $(a_0, \ldots a_{r-1})$ of elements in the field $\mathbb{F}
= \mathrm{GF}(2^{\max(n,m)})$, where $GF(2^k)$ denotes the Galois
field of $2^k$ elements.  Given $y \in \{0, 1\}^n$, the hash value
$h(y)$ can be computed by interpretting $y$ as an element of $F$,
computing $\Sigma_{j=0}^{r-1}a_jy^j$ in $F$, and selecting $m$ bits of
the encoding of the result.  The authors
of~\cite{Bellare98uniformgeneration} suggest polynomial-time
optimizations for operations in the field $F$.  Unfortunately, even
with these optimizations, computing algebraic hash functions is 
quite expensive in practice when non-linear terms are involved,
as in $\Sigma_{j=0}^{r-1}a_jy^j$,

%% Supratik:  But we need hash functions for r = 2 or 3.  So can
%% we still use algebraic hash functions?
Our approach uses computationally efficient linear hash functions.  As
we show later, pairwise independent hash functions
suffice for our purposes.  The literature describes several families
of efficiently computable pairwise independent hash functions.  One
such family, which we denote $H_{conv}(n, m, 2)$, is based on the
\emph{wrapped convolution} function~\cite{Mansour}.  For $a \in \{0,
1\}^{n+m-1}$ and $y \in \{0, 1\}^n$, the wrapped convolution $c = (a
\bullet y)$ is defined as an element of $\{0, 1\}^m$ as follows: for
each $i \in \{1, \ldots m\}$, $c[i] = \bigoplus_{j=1}^n (y[j] \wedge
a[i+j-1])$, where $\bigoplus$ denotes logical xor and $v[i]$ denotes
the $i^{th}$ component of the bit-vector $v$.  The family $H_{conv}(n,
m, 2)$ is defined as $\{h_{a,b}(y) = (a \bullet y) \oplus_m b \mid a
\in \{0, 1\}^{n+m-1}, b \in \{0, 1\}^m\}$, where $\oplus_m$ denotes
componentwise xor of two elements of $\{0, 1\}^m$.  By randomly
choosing $a$ and $b$, we can randomly choose a function $h_{a,b}(x)$
from this family.  It has been shown in~\cite{Mansour} that
$H_{conv}(n, m, 2)$ is pairwise independent.  Our implementation of a
near-uniform generator of CNF SAT witnesses uses $H_{conv}(n, m, 2)$.
%Another computationally efficient
%family, which we denote $H_{xor}(n, m, 3)$, is based on randomly
%choosing bits from $y \in \{0, 1\}^n$ and xor-ing them.  Let $h(y)[i]$
%denote the $i^{th}$ component of the bit-vector obtained by applying
%hash function $h$ to $y$.  The family $H_{xor}(n, m, 3)$ is defined as
%$\{h(y) \mid (h(y))[i] = a_{i,0} \oplus (\bigoplus_{k=1}^n
%a_{i,k}\cdot y[k]), a_{i,j} \in \{0, 1\}, 1 \leq i \leq m, 0 \leq j
%\leq n\}$.  By randomly choosing the $a_{i,j}$'s, we can randomly
%choose a hash function from this family.  As shown
%in~\cite{Gomes-Sampling}, the family $H_{xor}(n, m, 3)$ is $3$-wise
%(and hence also pair-wise) independent.

\section{Related Algorithms in Prior Work}\label{sec:bellare-algorithm}
We now discuss two algorithms that are closely related to our work.
In 1998, Bellare et al.~\cite{Bellare98uniformgeneration} proposed the
{\BGP} algorithm, showing that uniform generation of \mc{NP}-witnesses
can be achieved in probabilistic polynomial time using an
\mc{NP}-oracle.  This improved on previous work by Jerrum, Valiant and
Vazirani~\cite{Jerr}, who showed that uniform generation can be
achieved in probabilistic polynomial time using a $\Sigma_2^P$ oracle,
and almost-uniform generation (as defined in~\cite{Jerr}) can be
achieved in probabilistic polytime using an \mc{NP} oracle.

%In the
%remainder of the paper, we refer to the Bellare et al.'s algorithm as
%the {\BGP} algorithm (after the last names of the authors).

Let $R$ be an \mc{NP}-relation over $\Sigma$.  The {\BGP} algorithm
takes as input an $x \in L_R$ and either generates a witness that is
uniformly distributed in $R_x$, or produces a symbol $\bot$ (indicating a
failed run).  The pseudocode for the algorithm is presented below.  In
the presentation, we assume w.l.o.g. that $n$ is an integer such that $R_x
\subseteq \{0,1\}^n$.  We also assume access to \mc{NP}-oracles to
answer queries about cardinalities of witness sets and also to
enumerate small witness sets.
\begin{tabbing}	
xx \= xx \= xx \= xx \= xx \= \kill
\noindent {\bfseries \textsf{Algorithm} $\BGP(x):$}\\ 
/* Assume $R_x \subseteq \{0, 1\}^n$ */\\
1:\> $\mathrm{pivot} \leftarrow 2n^2$;\\
2:\> {\bfseries if} ($|R_x| \le \mathrm{pivot}$)\\
3:\> \> List all elements $y_1, \ldots y_{|R_x|}$ of $R_x$; \\ 
4:\> \> Choose $j$ at random from $\{1, \ldots |R_x|\}$, and {\bfseries
return} $y_j$;\\
5:\> {\bfseries else} \\
6:\> \> $l \leftarrow 2\lceil \log_2n \rceil$; $i \leftarrow \l-1$;\\
7:\> \> {\bfseries repeat}\\
8:\> \> \> $i \leftarrow i+1$;\\
9:\> \> \> Choose $h$ at random from $H(n, i-l, n)$; \\
10:\> \> {\bfseries until} ($\forall \alpha \in \{0,1\}^{i-l},|R_{x,h,\alpha}| \le 2n^{2}$) or ($i = n-1$);\\
11:\> \> {\bfseries if} ($\exists \alpha \in \{0,1\}^{i-l},|R_{x,h,\alpha}| > 2n^{2}$) {\bfseries return} $\bot$;\\
%12:\> \> \> $h \leftarrow h_{\mathrm{ID}}$, where \\
%\> \> \> $h_{\mathrm{ID}}: \{0, 1\}^n \rightarrow \{0, 1\}^{n-l}$ maps
%every
%$(u_1,\ldots u_n) \in \{0,1\}^n$ to $(u_1, \ldots u_{n-l})$;\\
12:\> \> Choose $\alpha$ at random from $\{0,1\}^{i-l}$;\\
13:\> \> List all elements $y_1, \ldots y_{|R_{x, h, \alpha}|}$ of $R_{x, h,
\alpha}$;\\
14:\> \> Choose $j$ at random from $\{1, \ldots \mathrm{pivot}\}$; \\
15:\> \> {\bfseries if} $j \le |R_{x, h, \alpha}|$, {\bfseries return}
$y_j$;\\
16:\> \> {\bfseries else} {\bfseries return} $\bot$;\\
\end{tabbing}	
\noindent 
For clarity of exposition, we have made a small adaptation to the
algorithm originally presented in~\cite{Bellare98uniformgeneration}.
Specifically, if $h$ does not satisfy ($\forall \alpha \in \{0,1\}^{i-l},|R_{x,h,\alpha}| \le 2n^{2}$) when the
loop in lines 7--10 terminates, the original algorithm forces a
specific choice of $h$.  Instead, algorithm {\BGP} simply outputs
$\bot$ (indicating a failed run) in this situation.  A closer look at
the analysis presented in~\cite{Bellare98uniformgeneration} shows that
all results continue to hold with this adaptation.  The authors
of~\cite{Bellare98uniformgeneration} use algebraic hash functions and
random choices of $n$-tuples in $\mathrm{GF}(2^{\max(n,i-l)})$ to
implement the selection of a random hash function in line 9 of the
pseudocode.  The following theorem summarizes the key properties of
the BGP algorithm ~\cite{Bellare98uniformgeneration}.
\begin{theorem}\label{theorem:bgp-bounds}
If a run of the {\BGP} algorithm is successful, the probability that
$y \in R_x $ is generated by the algorithm is independent of $y$.
Further, the probability that a run of the algorithm fails is $\leq 0.8$.
\end{theorem}
Since the probability of any witness $y \in R_x$ being output by a
successful run of the algorithm is independent of $y$, the {\BGP}
algorithm guarantees uniform generation of witnesses.  However, as we
argue in the next section, scaling the algorithm to even medium-sized
problem instances is quite difficult in practice. Indeed, we have
found no published report discussing any implementation of the
{\BGP} algorithm.

In 2007, Gomes et al.~\cite{Gomes-Sampling} presented two
closely related algorithms named {\XORSample} and {\XORSampleprime}
for near-uniform sampling of combinatorial spaces.  A key idea in both
these algorithms is to constrain a given instance $F$ of the CNF SAT
problem by a set of randomly selected xor constraints over the
variables appearing in $F$.  An xor constraint over a set $V$ of
variables is an equation of the form $e=c$, where $c\in\{0,1\}$ and
$e$ is the logical xor of a subset of $V$.  A probability distribution
$\mathbb{X}(|V|,q)$ over the set of all xor constraints over $V$ is
characterized by the probability $q$ of choosing a variable in $V$.  A
random xor constraint from $\mathbb{X}(|V|,q)$ is obtained by forming
an xor constraint where each variable in $V$ is chosen independently
with probability $q$, and $c$ is chosen uniformly at random.

We present the pseudocode of algorithm {\XORSampleprime} below.  The
algorithm uses a function {\SATModelCount} that takes a Boolean
formula $F$ and returns the exact count of witnesses of $F$. Algorithm
{\XORSampleprime} takes as inputs a CNF formula $F$, the parameter $q$
discussed above and an integer $s>0$.  Suppose the number of
variables in $F$ is $n$.  The algorithm proceeds by conjoining~$s$
xor constraints to~$F$, where the constraints are chosen randomly from
the distribution $\mathbb{X}(n, q)$.  Let $F'$ denote the conjunction
of $F$ and the random xor constraints, and let $mc$ denote the model
count (i.e., number of witnesses) of $F'$.  If $mc \ge 1$, the
algorithm enumerates the witnesses of $F'$ and chooses one witness at
random.  Otherwise, the algorithm outputs $\bot$, indicating a failed
run.
\begin{tabbing}	
xx \= xx \= xx \= xx \= xx \= \kill
\noindent {\bfseries \textsf{Algorithm} ${\XORSampleprime} (F, q, s)$}\\ 
/* $n =$ Number of variables in $F$ */\\
1:\> $Q_s \leftarrow \{s$  random xor constraints from
$\mathbb{X}(n,q)\}$;\\
2:\> $F' = F \wedge (\bigwedge_{f \in Q_s} f)$;\\
3:\> $mc \leftarrow \SATModelCount(F')$;\\
4:\> {\bfseries if} $(mc \ge 1)$ \\
5:\> \> Choose $i$ at random from $\{1, \ldots mc\}$;\\
6:\> \> List the first $i$ witnesses of $F'$;\\
7:\> \> {\bfseries return} $i^{th}$ witness of $F'$;\\
8:\> {\bfseries else} {\bfseries return} $\bot$;
\end{tabbing} 
Algorithm {\XORSample} can be viewed as a variant of algorithm
{\XORSampleprime} in which we check if $mc$ is exactly $1$ (instead of
$mc \ge 1$) in line 4 of the pseudocode.  An additional difference is
that if the check in line 4 fails, algorithm {\XORSample} starts
afresh from line 1 by randomly choosing $s$ xor constraints. In our
experiments, we observed that {\XORSampleprime} significantly
outperforms {\XORSample} in performance, hence we consider only
{\XORSampleprime} for comparison with our algorithm.  The following
theorem is proved in ~\cite{Gomes-Sampling}
%Let $\langle f_1, \ldots f_s \rangle$ denote the lexicographic
%ordering of the random xor constraints in $Q_s$.  Choosing the set
%$Q_s$ is equivalent to choosing a random hash function $h_{Q_s}: \{0,
%1\}^n \rightarrow \{0,1\}^s$, where $h_{Q_s}[i] = f_i$ for all $i \in
%\{1, \ldots s\}$.  In~\cite{Gomes-Sampling}, the authors showed that
%if $q = \frac{1}{2}$, the random hash function $h_{Q_s}$ is $3$-wise
%independent, i.e.  in $H_{xor}(n, s, 3)$.  This property was
%subsequently used in~\cite{Gomes-Sampling} to prove the following key
%theorem.

\begin{theorem}\label{theorem:gss-bounds}
Let $F$ be a Boolean formula with $2^{s^*}$ solutions.  Let $\alpha$
be such that $0 < \alpha < s^*$ and $s = s^* - \alpha$.  For a witness
$y$ of $F$, the probability with which {\XORSampleprime} with
parameters $q = \frac{1}{2}$ and $s$ outputs $y$ is bounded below by
$c'(\alpha) 2^{-s^*}$, where $c'(\alpha) = \frac{1-2^{-\alpha/3}}{(1 +
  2^{-\alpha})(1 + 2^{-\alpha/3})}$.  Further, {\XORSampleprime}
succeeds with probability larger than $c'(\alpha)$.
\end{theorem}
While the choice of $q = \frac{1}{2}$ allowed the authors
of~\cite{Gomes-Sampling} to prove Theorem~\ref{theorem:gss-bounds},
the authors acknowledge that finding witnesses of $F'$ is quite hard
in practice when random xor constraints are chosen from $\mathbb{X}(n,
\frac{1}{2})$.  Therefore, they advocate using values of $q$ much
smaller than $\frac{1}{2}$.  Unfortunately, the analysis that yields
the theoretical guarantees in Theorem~\ref{theorem:gss-bounds} does
not hold with these smaller values of $q$.  This illustrates the
conflict between witness generators with good performance in practice,
and those with good theoretical guarantees.
\vspace*{-0.15in}
\section{The UniWit Algorithm: Design and Analysis}\label{sec:our-algorithm}
We now describe an adaptation, called {\UniformWitness}, of the {\BGP} algorithm that scales to
much larger problem sizes than those that can be handled by the {\BGP}
algorithm, while weakening the guarantee of uniform generation to that
of near-uniform generation.  Experimental results, however, indicate
that the witnesses generated by our algorithm are fairly uniform in
practice.  Our algorithm can also be viewed as an adaptation of the
{\XORSampleprime} algorithm, in which we do not need to provide
hard-to-estimate problem-specific parameters like $s$ and $q$.  
%% Our
%% experience based on extensive experiments with an implementation of
%% {\XORSampleprime} shows that the algorithm works in practice only for
%% carefully chosen problem-specific values of the parameter $s$.  In
%% fact, the theoretical guarantees for {\XORSampleprime} also depend on
%% the choice of the parameter $s$.  Unfortunately, as shown in Section
%% ~\ref{sec:experiments}, estimating a reasonable value for $s$ takes
%% hours for most problem instances in our benchmark suite.  In contrast,
%% our algorithm does not require any problem-dependent parameter to be
%% provided as input.

%% As our
%% experiments demonstrate, our algorithm outperforms the
%% {\XORSampleprime} algorithm in practice, both in terms of run time and
%% uniformity of generated witnesses.

We begin with some observations about the {\BGP} algorithm. In what
follows, line numbers refer to those in the pseudocode of the {\BGP}
algorithm presented in Section~\ref{sec:bellare-algorithm}.  Our first
observation is that the loop in lines 7--10 of the pseudocode iterates
until either $|R_{x,h,\alpha}| \leq 2n^2$ for \emph{every} $\alpha \in
\{0, 1\}^{i-l}$ or $i$ increments to $n-1$.  Checking the first
condition is computationally prohibitive even for values of $i-l$ and $n$
as small as few tens.
So we ask if this condition can be simplified, perhaps with some
weakening of theoretical guarantees.  Indeed, we have found that if
the condition requires that $1 \le |R_{x, h, \alpha}| \le 2n^2$ for
a \emph{specific} $\alpha \in \{0,1\}^{i-l}$ (instead of for
\emph{every} $\alpha \in \{0, 1\}^{i-l}$), we can still guarantee
near-uniformity (but not uniformity) of the generated witnesses. This
suggests choosing both a random $h \in H(n, i-l, n)$ \emph{and} a
random $\alpha \in \{0,1\}^{i-l}$ \emph{within} the loop of lines
7--10.
%% The condition for exiting the loop could then simply check if
%% $1 \le |R_{x,h,\alpha}| \le 2n^2$ or if $i$ equals $n-1$.

The analysis presented in~\cite{Bellare98uniformgeneration} relies on
$h$ being sampled uniformly from a family of $n$-wise independent hash
functions.  In the context of generating SAT witnesses, $n$ denotes
the number of propositional variables in the input formula.  This can
be large (several thousands) in problems arising from directed random
testing.  Unfortunately, implementing $n$-wise independent hash
functions using algebraic hash functions (as advocated
in~\cite{Bellare98uniformgeneration}) for large values of $n$ is
computationally infeasible in practice.  This prompts us to ask if the
{\BGP} algorithm can be adapted to work with $r$-wise independent hash
functions for small values of $r$, and if simpler families of hash
functions can be used.  Indeed, we have found that with $r \ge 2$, an
adapted version of the {\BGP} algorithm can be made to generate
near-uniform witnesses.  We can also bound the probability of failure
of the adapted algorithm by a constant.  Significantly, the
sufficiency of pairwise independence allows us to use computationally
efficient xor-based families of hash functions, like $H_{conv}(n, m,
2)$ discussed in Section~\ref{sec:prelims-for-bellare}.  This
provides a significant scaling advantage to our algorithm vis-a-vis
the {\BGP} algorithm in practice.

In the context of uniform generation of SAT witnesses, checking if
$|R_x| \leq 2n^2$ (line 2 of pseudocode) or if $|R_{x,h,\alpha}| \leq
2n^2$ (line 10 of pseudocode, modified as suggested above) can be done
either by approximate model-counting or by repeated invokations of a
SAT solver.  State-of-the-art approximate model counting
techniques~\cite{Gomes06modelcounting} rely on randomly sampling the
witness space, suggesting a circular dependency.  Hence, we choose to
use a SAT solver as the back-end engine for enumerating and counting
witnesses.  Note that if $h$ is chosen randomly from $H_{conv}(n, m,
2)$, the formula for which we seek witnesses is the conjunction of the
original (CNF) formula and xor constraints encoding the inclusion of
each witness in $h^{-1}(\alpha)$.  We therefore choose to use a SAT
solver optimized for conjunctions of xor constraints and CNF clauses
as the back-end engine; specifically, we use CryptoMiniSAT (version
2.9.2)~\cite{CryptoMiniSAT}.

%During our experimentation, we observed that alternative approach of 
%using Tseitin encoding for the xor constraints and and using a CNF SAT 
%solver like miniSAT or picoSAT ~\cite{minisat,picosat} couldn't scale to 
%even small problem instances.
%using a CNF SAT solver like miniSAT or picoSAT~\cite{minisat,picosat}
%Our experiments indicate that this leads to much better scalability of the
%algorithm vis-a-vis using Tseitin encoding for the xor constraints and
%using a CNF SAT solver like miniSAT or picoSAT~\cite{minisat,picosat}.
%% KSM : Do we have to really give justification of not using any other solver?

Modern SAT solvers often produce partial assignments that specify
values of a subset of variables, such that every assignment of values
to the remaining variables gives a witness.  Since we must find large
numbers ($2n^2 \approx 2\times 10^6$ if $n \approx 1000$) of
witnesses, it would be useful to obtain partial assignments from the
SAT solver.  Unfortunately, conjoining random xor constraints to the
original formula reduces the likelihood that large sets of witnesses
can be encoded as partial assignments.  Thus, each invokation of the
SAT solver is likely to generate only a few witnesses, necessitating a
large number of calls to the solver.  To make matters worse, if the
count of witnesses exceeds $2n^2$ and if $i < n-1$, the check in line
10 of the pseudocode of algorithm {\BGP} (modified as suggested above)
fails, and the loop of lines 7--10 iterates once more, requiring
generation of up to $2n^2$ witnesses of a modified SAT problem all
over again.  This can be computationally prohibitive in practice.
Indeed, our implementation of the {\BGP} algorithm with CryptoMiniSAT
failed to terminate on formulas with few tens of variables, even
when running on high-performance computers for $20$ hours.
%myv: I don't think this comment is needed.
%\footnote{The time limit is due to centralized cluster resource
% allocation policies.}  
This prompts us to ask if the required number of witnesses, or
\emph{pivot}, in the {\BGP} algorithm (see line 1 of the pseudocode)
can be reduced.  We answer this question in the affirmative, and
show that the pivot can indeed be reduced to $2n^{1/k}$, where 
$k$ is an integer $\ge 1$.  Note that if
$k=3$ and $n = 1000$, the value of $2n^{1/k}$ is only
$20$, while $2n^2$ equals $2\times10^6$. This
translates to a significant leap in the sizes of problems for which we
can generate random witnesses.  There are, however, some practical tradeoffs
involved in the choice of $k$; we defer a discussion of these to a
later part of this section.
 
We now present the {\UniformWitness} algorithm, which implements the
modifications to the {\BGP} algorithm suggested above.
{\UniformWitness} takes as inputs a CNF formula~$F$ with
$n$~variables, and an integer $k \ge 1$.  
%The algorithm uses one of
%the families $H_{conv}(n, m, 2)$ or $H_{xor}(n, m, 3)$ for randomly
%choosing hash functions.  However, the choice of the family of hash
%functions is assumed to remain the same for all runs of the algorithm.
The algorithm either outputs a witness that is near-uniformly
distributed over the space of all witnesses of $F$ or produces a 
symbol $\bot$ indicating a failed
run.  We also assume that we have access to a function {\BoundedSAT}
that takes as inputs a propositional formula $F$ that is a conjunction of a
CNF formula and xor constraints, and an integer $r \ge 0$ and returns
a set $S$ of witnesses of $F$ such that $|S| = \min(r, \#F)$, where
$\#F$ denotes the count of all witnesses of $F$.
%\vspace*{-0.1in}
\begin{tabbing}	
xx \= xx \= xx \= xx \= xx \= \kill
\noindent {\bfseries \textsf{Algorithm} ${\UniformWitness} (F, k)$:}\\ 
/* Assume $z_1, \ldots z_n$ are variables in $F$    */ \\
/* Choose a priori the family of hash functions $H(n, m, r), r \ge 2$ to be used  */\\\\
1:\> $\mathrm{pivot} \leftarrow \lceil 2n^{1/k} \rceil$; $S \leftarrow \BoundedSAT(F, \mathrm{pivot} +1)$;\\
2:\> {\bfseries if} ($|S| \le \mathrm{pivot}$)\\
3:\> \> Let $y_1, \ldots y_{|S|}$ be the elements of $S$;\\ 
4:\> \> Choose $j$ at random from $\{1, \ldots |S|\}$ and {\bfseries return} $y_j$;\\
5:\> {\bfseries else} \\
6:\> \> $l \leftarrow \lfloor \frac{1}{k}\cdot(\log_2 n)\rfloor$; $i \leftarrow l-1$;\\
7:\> \> {\bfseries repeat}\\
8:\> \> \> $i \leftarrow i+1$;\\
9:\> \> \> Choose $h$ at random from $H(n, i-l, r)$; \\
10:\> \> \> Choose $\alpha$ at random from $\{0, 1\}^{i-l}$;\\
11:\> \> \> $S \leftarrow \BoundedSAT(F \wedge (h(z_1, \ldots z_n) = \alpha), \mathrm{pivot}+1)$;\\
12:\> \> {\bfseries until} ($1 \le |S| \le \mathrm{pivot}$) or ($i = n$);\\
13:\> \> {\bfseries if} ($|S| > \mathrm{pivot}$) or ($|S| < 1$) {\bfseries return} $\bot$;\\
14:\> \> {\bfseries else}\\
15:\> \> \> Let $y_1, \ldots y_{|S|}$ be the elements of $S$;\\
16:\> \> \> Choose $j$ at random from $\{1, \ldots \mathrm{pivot}\}$;. \\
17:\> \> \> {\bfseries if} $j \le |S|$, {\bfseries return} $y_j$;\\
18:\> \> \> {\bfseries else} {\bfseries return} $\bot$;\\
\end{tabbing}

\noindent {\bfseries Implementation issues:} 
There are four steps in {\UniformWitness} (lines 4, 9, 10 and 16 of
the pseudocode) where random choices are made.  In our implementation,
in line 10 of the pseudocode, we choose a random hash function from
the family $H_{conv}(n,i-l,2)$, since it is computationally efficient
to do so.
%either $H_{conv}(n, i-l, 2)$ or $H_{xor}(n, i-l,3)$. 
Recall from Section~\ref{sec:prelims-for-bellare} that choosing a
random hash function from $H_{conv}(n, m, 2)$ requires choosing two
random bit-vectors.  It is straightforward to implement these choices
and also the choice of a random $\alpha \in \{0, 1\}^{i-l}$ in line 10
of the pseudocode, if we have access to a source of independent and
uniformly distributed random bits.  In lines 4 and 16, we must choose
a random integer from a specified range.  By using standard techniques
(see, for example, the discussion on coin tossing
in~\cite{Bellare98uniformgeneration}), this can also be implemented
efficiently if we have access to a source of random bits.  Since
accessing truly random bits is a practical impossibility, our
implementation uses pseudorandom sequences of bits generated from
nuclear decay processes and available at HotBits~\cite{HotBits}.  We
download and store a sufficiently long sequence of random bits in a
file, and access an appropriate number of bits sequentially whenever
needed.

%While we could use either $H_{conv}(n, i-l, 2)$ or $H_{xor}(n, i-l,
%3)$ in line 9 of the pseudocode, our implementation uses $H_{conv}(n,
%i-l, 2)$, since we found this to be more efficient in practice than
%working with hash functions from $H_{xor}(n, i-l, 3)$.
%While we could use any hash function family in line 9 of the pseudocode, our implementation uses computationally efficient $H_{conv}(n, i-l, 2)$ 
% in line 9 of the pseudocode.
%Section~\ref{sec:our-algorithm-analysis}, we present a formal analysis
%proving near uniform generation of witnesses for our algorithm.   Significantly,
%unlike algorithm {\XORSampleprime}, we can bound the probability of failure of our algorithm 
%by a constant independent of the choice of any
%problem-specific parameter.

%%  in a loop that iterates until either $r$ witnesses of
%% $F$ have been generated or no new witnesses (beyond those already
%% generated) exist.  Suppose an invokation of CryptoMiniSAT returns a
%% (possibly partial) assignment of $m (\le n)$ variables in $F$.  We
%% generate all $2^{n-m}$ witnesses corresponding to this partial
%% assignment, add $2^{n-v}$ to the number of witnesses generated so far,
%% and add a blocking clause obtained by negating the partial assignment
%% to $F$, before iterating through the loop again.

In line 11 of the pseudocode for {\UniformWitness}, we invoke
{\BoundedSAT} with arguments $F\wedge (h(z_1, \ldots z_n) = \alpha)$
and $\mathrm{pivot}+1$.  The function {\BoundedSAT} is implemented
using CryptoMiniSAT (version 2.9.2), which allows passing a parameter indicating the
maximum number of witnesses to be generated.  The sub-formula $(h(z_1,
\ldots z_n) = \alpha)$ is constructed as follows.  As mentioned in
Section~\ref{sec:prelims-for-bellare}, a random hash function from the
family $H_{conv}(n, i-l, 2)$ can be implemented by choosing a random
$a \in \{0, 1\}^{n+i-l-1}$ and a random $b \in \{0, 1\}^{i-l}$.  %% Let
%% $a$ and $b$ be two such random bit-vectors and let $h$ be the
%% corresponding hash function. 
Recalling the definition of $h$ from
Section~\ref{sec:prelims-for-bellare}, the sub-formula $(h(z_1, \ldots
z_n) = \alpha)$ is given by
$\bigwedge_{j=1}^{i-l}\left(\left(\bigoplus_{p=1}^{n}(z_p \wedge
a[j+p-1]) \oplus b[j]\right) \Leftrightarrow \alpha[j]\right)$.
\noindent {\bfseries Analysis of {\UniformWitness}:}
Let $R_F$ denote the set of witnesses of the input formula $F$.  Using
notation discussed in Section~\ref{sec:prelims-for-bellare}, suppose
$R_F \subseteq \{0, 1\}^n$.  For simplicity of exposition, we assume
that $\log_2|R_F|-(1/k)\cdot\log_2n$ is an integer in the following
discussion.  A more careful analysis removes this assumption with
constant factor reductions in the probability of generation of an
arbitrary witness and in the probability of failure of
{\UniformWitness}.

\begin{theorem}\label{theorem:uniform-gen}
Suppose $F$ has $n$ variables and $n > 2^k$.  For every witness $y$
of $F$, the conditional probability that algorithm {\UniformWitness}
outputs $y$ on inputs $F$ and $k$, given that the algorithm succeeds,
is bounded below by $\frac{1}{8|R_F|}$.
\end{theorem}
\begin{proof} Referring to the pseudocode of {\UniformWitness}, if
$|R_F| \le 2n^{1/k}$, the theorem holds trivially.  Suppose $|R_F| >
2n^{1/k}$, and let $Y$ denote the event that witness $y$ in $R_F$ is
output by {\UniformWitness} on inputs $F$ and $k$.  Let $p_{i,y}$
denote the probability that the loop in lines 7--12 of the pseudocode
terminates in iteration $i$ with $y$ in $R_{F,h,\alpha}$, where
$\alpha \in \{0, 1\}^{i-l}$ is the value chosen in line 10.  It
follows from the pseudocode that $\prob\left[Y\right] \geq
p_{i,y}\cdot(1/2n^{1/k})$, for every $i \in \{l, \ldots n\}$.  Let us
denote $\log_2|R_F|-(1/k)\cdot\log_2n$ by $m$.  Therefore, $2^m\cdot
n^{1/k} = |R_F|$.  Since $2n^{1/k} < |R_F| \leq 2^n$ and since $l
= \lfloor(1/k)\cdot\log_2 n\rfloor$ (see line 6 of pseudocode), we have $l < m+l \leq n$.
Consequently, $\prob\left[Y\right] \geq p_{m+l,y}\cdot(1/2n^{1/k})$.
The proof is completed by showing that $p_{m+l,y} \geq \frac{1 -
n^{-1/k}}{2^{m+1}}$.  This gives
$\prob\left[Y\right] \geq \frac{1-n^{-1/k}}{2^{m+2}\cdot n^{1/k}}$ $=$
$\frac{1-n^{-1/k}}{4|R_F|}$ $\ge \frac{1}{8|R_F|}$, if $n > 2^k$.

To calculate $p_{m+l,y}$, we first note that since $y \in R_F$, the
requirement ``$y \in R_{F,h,\alpha}$'' reduces to ``$y \in h^{-1}(\alpha)$''.
For $\alpha \in \{0,1\}^m$ and $y \in \{0,1\}^n$, we define
$q_{m+l,y,\alpha}$ as $\prob\left[|R_{F,h,\alpha}| \leq
2n^{1/k} \mbox{ and } h(y) = \alpha : h \xleftarrow{R} H(n, m,
r)\right]$, where $r \ge 2$.  The proof is now completed by showing
that $q_{m+l,y,\alpha} \geq (1 - n^{-1/k})/2^{m+1}$ for every
$\alpha \in \{0,1\}^m$ and $y \in \{0,1\}^n$.  Towards this end, we
define an indicator variable $\gamma_{y, \alpha}$ for every
$y \in \{0,1\}^n$ and $\alpha
\in \{0,1\}^m$ as follows: $\gamma_{y,\alpha} = 1$ if $h(y) = \alpha$
and $\gamma_{y,\alpha} = 0$ otherwise.  Thus, $\gamma_{y, \alpha}$ is
a random variable with probability distribution induced by that of
$h$.  It is easy to show that (i)
$\expect\left[\gamma_{y,\alpha}\right] = 2^{-m}$, and (ii) the
pairwise independence of $h$ implies pairwise independence of the
$\gamma_{y, \alpha}$ variables.  We now define $\Gamma_\alpha
= \sum_{z \in R_F} \gamma_{z, \alpha}$ and $\mu_{y, \alpha} =
\expect\left[\Gamma_\alpha \mid \right.$ $\left.  \gamma_{y, \alpha} =
  1\right]$.  Clearly, $\Gamma_\alpha = |R_{F, h, \alpha}|$ and
$\mu_{y, \alpha} = \expect\left[\sum_{z \in R_F}\gamma_{z, \alpha}
  \mid \gamma_{y,\alpha} = 1 \right]$ $= \sum_{z \in R_F}
\expect\left[\gamma_{z, \alpha} \mid \gamma_{y, \alpha} = 1\right]$.
Using pairwise independence of the $\gamma_{y, \alpha}$ variables, the
above simplifies to $\mu_{y,\alpha} = 2^{-m}(|R_F| - 1) + 1$ $\le
2^{-m}|R_F|+1 = n^{1/k}+1$.  From Markov's inequality, we know that
$\prob \left[\,\Gamma_\alpha \le \kappa\cdot\mu_{y,\alpha} \right.$
$\left.  \mid \gamma_{y,\alpha} = 1 \right]$ $\ge 1 - 1/\kappa$ for
$\kappa > 0$.  With $\kappa = \frac{2}{1+n^{-1/k}}$, this gives
$\prob\left[\,|R_{F, h, \alpha}| \right.$ $\left. \le
2n^{1/k} \right.$ $\left. \mid \gamma_{y,\alpha} = 1\right]$ $\ge
(1-n^{-1/k})/2$.  Since $h$ is chosen at random from $H(n, m, r)$, we
also have $\prob\left[h(y) = \alpha\right] = 1/2^m$.  It follows that
$q_{m+l,y,\alpha} \ge (1-n^{-1/k})/2^{m+1}$. \qed
\end{proof}      

%	Next we show that for 3 or more independence we can get stronger lower bound.
	
%	If $h \in H(n,m,3)$, then $\sigma^2_{y',\alpha} = \sum_{y \neq y', y\in
%	  R_x}\var\left[\gamma_{y,\alpha}\right]$ $\le \sum_{y \neq y', y\in
%	  R_x}\expect\left[\gamma_{y,\alpha}\right]$ $ \le \sum_{y \in
%	  R_x}\expect\left[\gamma_{y, \alpha}\right]$
%	$=\expect\left[\Gamma_\alpha\right] = 2^{-m}|R_x|$.
%	From Chebyshev's inequality, we know that $\prob\left[\,|\Gamma_\alpha
%	  - \mu_{y',\alpha}| \ge\right.$ $\left.\kappa\sigma_{y',\alpha}  \right.$  $\left. \mid
%	  \gamma_{y',\alpha} = 1\right]$ $\leq \frac{1}{\kappa^2}$ for every
%	$\kappa > 0$.  By choosing $\kappa =
%	\frac{\mu_{y',\alpha}}{2\sigma_{y',\alpha}}$, and substituting the
%	  expressions for $\Gamma_\alpha$ and $\mu_{y',\alpha}$, we get
%	  $\prob\left[\,|R_{x, h, \alpha}| \le \frac{|R_x|}{2^{m+1}}\right.$
%	    or $\left.|R_{x, h, \alpha}| \ge \frac{3|R_x|}{2^{m+1}} \mid
%	    \gamma_{y', \alpha} = 1\right]$ $\le \frac{4}{2^{-m}|R_x|}$.
%	  Since $n^{1/k} = 2^{-m}|R_x|$, we have
%	   \[\prob\left[\,(|R_{x,
%	      h,\alpha}|\le \frac{1}{2}n^{1/k}) \mbox{ or } (|R_{x,
%	      h,\alpha}| \ge \frac{3}{2}n^{1/k}) \mid \gamma_{y',\alpha} =
%	    1\right]  \le 4n^{-1/k}\]
%	 It is easy
%	to see that $q_{m+l,y} \ge $ $\prob\left[\,\frac{1}{2}n^{1/k} \le
%	  |R_{x,h,\alpha}|\le \frac{3}{2}n^{1/k} \mid \gamma_{y',\alpha} =
%	  1\right]\cdot\prob\left[\gamma_{y',\alpha} = 1\right]$ $\ge
%	\frac{1-4n^{-1/k}}{2^{m}}$. Hence $\prob\left[Y\right]
%	\ge \frac{1}{2n^{1/k}}\frac{1-4n^{-1/k}}{2^{m}}$.
\noindent

\begin{theorem}\label{theorem:fail-prob}
Assuming $n > 2^k$, algorithm {\UniformWitness} succeeds (i.e. does
not return $\bot$) with probability at least $\frac{1}{8}$.
\end{theorem}
\begin{proof} Let $P_{succ}$ denote the probability that a run of
algorithm {\UniformWitness} succeeds.  By definition,
$P_{succ} = \sum_{y \in R_F} \prob\left[Y\right]$. Using Theorem
\ref{theorem:uniform-gen}, $P_{succ} \ge \sum_{y \in R_F} \frac{1}{8|R_F|}$ 
$= \frac{1}{8}$.\qed
\end{proof}
\noindent

One might be tempted to use large values of the parameter $k$ to keep
the value of $pivot$ low.  However, there are tradeoffs involved in
the choice of $k$.
%% Since we do not know $|R_x|$ a priori, the loop in lines 8--13 of the
%% pseudocode needs to iterate over values of $i$ starting from
%% $\frac{1}{k}\log_2 n$ and increasing up to $n$ (in the worst-case).
%% Therefore, as $k$ increases, the starting value for $i$ reduces and in
%% the worst-case, more iterations are needed.  Increasing $k$ also has
%% another interesting effect.  
%% In each iteration of the loop in lines 7--12 of the pseudocode, if
%% {\BoundedSAT} returns more than $2n^{1/k} + 2$ witnesses, we increment
%% $i$ and iterate through the loop once again. 
As $k$ increases, the pivot $2n^{1/k} $ reduces, and the chances that 
{\BoundedSAT} finds more than $2n^{1/k} $ witnesses increases, 
necessitating further iterations of the loop in lines 7--12 of the pseudocode.
Of course, reducing the 
pivot also means that {\BoundedSAT} has to find fewer witnesses, and each 
invokation of
{\BoundedSAT} is likely to take less time.  However, the increase in the
number of invokations of {\BoundedSAT} contributes to
increased overall time.  In our experiments, we have found that
choosing $k$ to be either $2$ or $3$ works well for all our benchmarks
(including those containing several thousand variables).  

\noindent {\bfseries A heuristic optimization:} A (near-)uniform generator
is likely to be invoked a large number of times for the same formula
$F$ when generating a set of witnesses of $F$.  If the performance of
the generator is sensitive to problem-specific parameter(s) not known
a priori, a natural optimization is to estimate values of these
parameter(s), perhaps using computationally expensive techniques, in
the first few runs of the generator, and then re-use these estimates
in subsequent runs on the same problem instance.  Of course, this
optimization works only if the parameter(s) under consideration can be
reasonably estimated from the first few runs.  We call this heuristic
optimization ``leapfrogging''.

In the case of algorithm {\UniformWitness}, the loop in lines 7--12 of
the pseudocode starts with $i$ set to $l-1$ and iterates until either
$i$ increments to $n$, or $|R_{F,h,\alpha}|$ becomes no larger than
$2n^{1/k}$.  For each problem instance $F$, we propose to estimate a
lower bound of the value of $i$ when the loop terminates, from the
first few runs of {\UniformWitness} on $F$.  In all subsequent runs of
{\UniformWitness} on $F$, we propose to start iterating through the
loop with $i$ set to this lower bound.  We call this specific
heuristic \emph{``leapfrogging $i$''} in the context of
{\UniformWitness}.
%% Our analysis of {\UniformWitness} shows
%% that the probabilistic guarantees of {\UniformWitness} continue to
%% hold as long as the lower bound of $i$ used in leapfrogging is small
%% enough.  
Note that leapfrogging may also be used for the parameter
$s$ in algorithms {\XORSampleprime} and {\XORSample} (see 
pseudocode of {\XORSampleprime}).  We will discuss more about 
this in Section~\ref{sec:experiments}.
\vspace*{-0.15in}  
			 
\section{Experimental Results}\label{sec:experiments}
To evaluate the performance of {\UniformWitness}, we built a prototype
implementation and conducted an extensive set of experiments.  Since
our motivation stems primarily from functional verification, our
benchmarks were mostly derived from functional verification of
hardware designs.  Specifically, we used ``bit-blasted''
versions of word-level constraints arising from bounded model checking
of public-domain and proprietary word-level VHDL designs.  In
addition, we also used bit-blasted versions of several
SMTLib~\cite{SMTLib} benchmarks of the {\small
  ``QF\_BV/bruttomesso/} {\small simple\_processor/''} category, and
benchmarks arising from ``Type I'' representations of
ISCAS'85 circuits, as described in~\cite{Darwiche02acompiler}.
%Finally, we considered $100$ hard random 3-SAT instances, with the
%variable and clause count fixed a priori. 
%$KSM$: We are not presenting results
%with random sat, although we can write it here and have results on the web page

%% Generation of random instance was through generation of 100 uniform
%% random 3-SAT instances with variables and clauses fixed a
%% priori. \cite{Mitchell92} 

All our experiments were conducted on a high-performance computing
cluster.  Each individual experiment was run on a single node of the
cluster, and the cluster allowed multiple experiments to run in
parallel.  Every node in the cluster had two quad-core Intel Xeon
processors running at $2.83$ GHz with $4$ GB of physical memory.  We
used $3000$ seconds as the timeout interval for each invokation of
{\BoundedSAT} in {\UniformWitness}, and $20$ hours as the timeout
interval for the overall algorithm.  If an invokation of {\BoundedSAT}
in line 11 of the pseudocode timed out (after $3000$ seconds), we
repeated the iteration (lines 7--12 of the pseudocode of
{\UniformWitness}) without incrementing $i$.  If the overall algorithm
timed out (after $20$ hours), we considered the algorithm to have
failed.  We used either $2$ or $3$ for the value of the parameter $k$
(see pseudocode of {\UniformWitness}).  This corresponds to
restricting the pivot to few tens of witnesses for formulae with a few
thousand variables.  The exact values of $k$ used for a subset of the
benchmarks are indicated in Table~\ref{table:comparison}.  A full
analysis of the effect of parameter $k$ will require a
separate study.  As explained earlier, our implementation uses the
family $H_{conv}(n, m, 2)$ to select random hash functions in step 9
of the pseudocode.
             
%\vspace*{-0.2in}		
%\vspace*{-0.15in}

For purposes of comparison, we also implemented and conducted
experiments with algorithms {\BGP}~\cite{Bellare98uniformgeneration},
{\XORSample} and {\XORSampleprime}~\cite{Gomes-Sampling}, using
CryptoMiniSAT as the SAT solver in all cases.  Algorithm {\BGP} timed
out without producing any witness in all but the simplest of cases
(involving less than $20$ variables).  This is primarily because
checking whether $|R_{x,h,\alpha}| \leq 2n^2$ for a given $h \in H(n,
m, n)$ and \emph{for every} $\alpha \in \{0, 1\}^m$, as required in
step $10$ of algorithm {\BGP}, is computationally prohibitive for
values of $n$ and $m$ exceeding few tens.  Hence, we do not report any
comparison with algorithm {\BGP}.  Of the algorithms {\XORSample} and
{\XORSampleprime}, algorithm {\XORSampleprime} consistently
out-performed algorithm {\XORSample} in terms of both actual time
taken and uniformity of generated witnesses.  This can be largely
attributed to the stringent requirement that algorithm {\XORSample} be
provided a parameter $s$ that renders the model count of the input
formula $F$ constrained with $s$ random xor constraints to
\emph{exactly} $1$.  Our experiments indicated that it was extremely
difficult to predict or leapfrog the range of values for $s$ such that
it met the strict requirement of the model count being \emph{exactly}
$1$.  This forced us to expend significant computing resources to
estimate the right value value for $s$ in almost every run, leading to
huge performance overheads.  Since algorithm {\XORSampleprime}
consistently outperformed algorithm {\XORSample}, we focus on
comparisons with only algorithm {\XORSampleprime} in the subsequent
discussion.  \added{Note that our benchmarks, when viewed as Boolean
  circuits, had upto $695$ circuit inputs, and $21$ of them had more
  than $95$ inputs each.  While {\UniformWitness} and
  {\XORSampleprime} completed execution on all these benchmarks, we
  could not build ROBDDs for $18$ of the above $21$ benchmarks within
  our timeout limit and with 4GB of memory.}

%Whilethe authors of~\cite{Gomes-Sampling} describe two algorithms, we used
%{\XORSampleprime} for comparison, since this is closer to our
%algorithm.
%%Kuldeep : We have mentioned this in Sectioni 3, so removing it from here

%The table \ref{table:comparison} compares performance of
%{\UniformWitness} and {\XORSampleprime} on various benchmarks in
%terms of Time taken by each sampler to return first sample and the
%comparison of Kullback-Leibler(KL) divergence between the data from
%each sampler and the uniform distribution. Here we consider real
%industry use case where a sufficiently large sample but small
%compared to the total solutions is required and we are interested in
%the KL divergence of the sample. Since this still required
%sufficiently large experiments, hence for larger benchmarks the
%results for KL divergence are not presented due to computation
%limitations of our resources.  The results for mean KL divergence
%were obtained from KL divergence for each of the 60 runs for each
%benchmark where every run produced 1000 samples. The comparison
%between two samplers shows the dramatic difference in the quality of
%samples produced by each sampler. \\

Table~\ref{table:comparison} presents results of our experiments
comparing performance and uniformity of generated witnesses for
{\UniformWitness} and {\XORSampleprime} on a subset of benchmarks.
\added{The tool and the complete set of results on over $200$
  benchmarks} are available at \url{
  http://www.cfdvs.iitb.ac.in/reports/reports/CAV13/}.
%scriptsize
\begin{table}[h]
\begin{center}	
\begin{tabular}{|c| c| c| c| c| c| c|| c| c|}
\hline
\multicolumn{3}{|c|}{}& \multicolumn{4}{|c|}{\UniformWitness} &
\multicolumn{2}{|c|}{\XORSampleprime} \\
\hline 
Benchmark & \#var & Clauses & k& \shortstack{Range\\(i)} & \shortstack{Average\\Run Time (s)}
&\shortstack{Var-\\iance}  & \shortstack{Average\\Run Time (s)} & \shortstack{Var-\\iance}\\
\hline
%% \multirow{2}{*}{case47}&\multirow{2}{*}{118}&\multirow{2}{*}{328}&2&[14,15]&0.5+0.19&\multirow{2}{*}{5.92}&\multirow{2}{*}{1133.91+1133.91}&\multirow{2}{*}{18.13}\\ 
 
%% &&&3&[14,17]&0.41+0.24&&& \\ 
%%  \hline
%%  \multirow{2}{*}{case8}&\multirow{2}{*}{183}&\multirow{2}{*}{593}&2&[18,20]&4.86+2.04&\multirow{2}{*}{0.98}&\multirow{2}{*}{9370.78+369.21}&\multirow{2}{*}{0.65}\\ 
 
%% &&&3&[19,22]&4.21+1.9&&& \\ 
%%  \hline
%% \multirow{2}{*}{case\_3\_b14}&\multirow{2}{*}{1293}&\multirow{2}{*}{4200}&2&[31,32]&49.54+4.07&\multirow{2}{*}{10.88}&\multirow{2}{*}{15004.6+2.08}&\multirow{2}{*}{13.82}\\ 
 
%% &&&3&[32,34]&17.41+1.21&&& \\ 
%%  \hline
 \multirow{2}{*}{case\_3\_b14}&\multirow{2}{*}{779}&\multirow{2}{*}{2480}&2&[34,35]&49.29+5.27&\multirow{2}{*}{1.58}&\multirow{2}{*}{15061.85+59.31}&\multirow{2}{*}{3.47}\\ 
 
&&&3&[36,37]&19.32+1.44&&& \\ 
\hline
 case\_2\_b14&519&1607&3&[38,39]&22.13+2.09&0.57&18005.58+0.73&9.51\\ 
 \hline
case203&214&580&3&[42,44]&16.41+1.04&8.98&18006.85+2.78&230.5\\ 
 \hline
case145&219&558&3&[42,44]&19.84+1.42&1.62&18007.18+2.99&2.32\\ 
 \hline
case14&270&717&2&[44,45]&54.07+2.33&0.65&18004.8+0.9&28.16\\ 
 \hline
case61&289&773&3&[44,46]&30.39+5.49&1.33&18009.1+4.4&11.92\\ 
 \hline
case9&302&821&3&[45,47]&25.64+1.54&2.07&18004.79+0.87&46.15\\ 
 \hline
case10&351&946&2&[60,61]&204.93+17.99&6.1&18008.42+4.85&10.56\\ 
 \hline
 case15&319&842&3&[61,63]&91.84+14.64&0.82&18008.34+5.08&11.04\\ 
 \hline
case140&488&1222&3&[99,101]&288.63+23.53&3.4&21214.85+200.64&6.71\\ 
 \hline
squaring14&5397&18141&3&[28,30]&2399.19+1243.81&&7089.6+2088.46&\\ 
 \hline
squaring7&5567&18969&3&[26,29]&2358.45+1720.49&&4841.4+2340.84&\\ 
 \hline
case39&590&1789&2&[50,50]&710.65+85.22&&18159.12+138.22&\\ 
 \hline
case\_2\_ptb&7621&24889&3&[72,73]&1643.2+225.41&&22251.8+177.61&\\ 
 \hline
\multirow{2}{*}{case\_1\_ptb}&\multirow{2}{*}{7624}&\multirow{2}{*}{24897}&2&[70,70]&17295.45+454.64&\multirow{2}{*}{}&\multirow{2}{*}{22346.64+204.07}&\multirow{2}{*}{}\\ 
 
&&&3&[72,73]&1639.16+219.87&&& \\ 
 \hline
\end{tabular}
\end{center}
\caption{Performance comparison of {\UniformWitness} and {\XORSampleprime}}
\label{table:comparison}
\end{table}	
The first three columns in Table~\ref{table:comparison} give the name,
number of variables and number of clauses of the benchmarks
represented as CNF formulae.   The columns grouped
under {\UniformWitness} give details of runs of {\UniformWitness},
while those grouped under {\XORSampleprime} give details of runs of
{\XORSampleprime}.  For runs of {\UniformWitness}, the column labeled
``$k$'' gives the value of the parameter $k$ used in the corresponding
experiment.  The column labeled ``Range ($i$)'' shows the range of
values of $i$ when the loop in lines 7--12 of the pseudocode (see
Section~\ref{sec:our-algorithm}) terminated in $100$ independent runs
of the algorithm on the benchmark under consideration.  Significantly,
this range is uniformly narrow for all our experiments with
{\UniformWitness}.  As a result, leapfrogging $i$ is very effective
for {\UniformWitness}.

The column labeled ``Run Time'' under {\UniformWitness} in
Table~\ref{table:comparison} gives run times in seconds, separated as
$time_1 + time_2$, where $time_1$ gives the average time (over $100$
independent runs) to obtain a witness and to identify the lower bound
of $i$ for leapfrogging in later runs, while $time_2$ gives the
average time to get a solution once we leapfrog $i$.  Our experiments
clearly show that leapfrogging $i$ reduces run-times by almost an
order of magnitude in most cases.  We also report ``Run Time'' for
{\XORSampleprime}, where times are again reported as $time_1 +
time_2$.  In this case, $time_1$ gives the average time (over $100$
independent runs) taken to find the value of the parameter $s$ in
algorithm {\XORSampleprime} using a binary search technique, as
outlined in a footnote in~\cite{Gomes-Sampling}.  As can be seen from
Table~\ref{table:comparison}, this is a computationally expensive
step, and often exceeds $time_1$ under {\UniformWitness} by more than
two to three orders of magnitude.  Once the range of the parameter $s$
is identified from the first $100$ independent runs, we use the lower
bound of this range to leapfrog $s$ in subsequent runs of
{\XORSampleprime} on the same problem instance. The values of $time_2$
under ``Run Time'' for {\XORSampleprime} give the average time taken
to generate witnesses after leapfrogging $s$.  Note that the
difference between $time_2$ values for {\UniformWitness} and
{\XORSampleprime} algorithms is far less pronounced than the
difference between $time_1$ values.  In addition, the $time_1$ values
for {\XORSampleprime} are two to four orders of magnitude larger than
the corresponding $time_2$ values, while this factor is almost always
less than an order of magnitude for {\UniformWitness}. Therefore, the
total time taken for $n_1$ runs without leapfrogging, followed by
$n_2$ runs with leapfrogging for {\XORSampleprime} far exceeds that
for {\UniformWitness}, even for $n_1 = 100$ and $n_2 \approx 10^6$.
This illustrates the significant practical efficiency of
{\UniformWitness} vis-a-vis {\XORSampleprime}.

Table~\ref{table:comparison} also reports the scaled statistical
variance of relative frequencies of witnesses generated by $5 \times
10^4$ runs of the two algorithms on several benchmarks.  The scaled
statistical variance is computed as $\frac{K}{N-1}\sum\limits_{i=1}^N
\left(f_i-\left(\frac{\sum_{i=1}^N f_i}{N}\right)\right)^2$, where $N$
denotes the number of distinct witnesses generated, $f_i$ denotes the
relative frequency of the $i^{th}$ witness, and $K ~( 10^{10})$
denotes a scaling constant used to facilitate easier comparison.  The
smaller the scaled variance, the more uniform is the generated
distribution.  Unfortunately, getting a reliable estimate of the
variance requires generating witnesses from runs that sample the
witness space sufficiently well.  While we could do this for several
benchmarks (listed towards the top of Table~\ref{table:comparison}),
other benchmarks (listed towards the bottom of
Table~\ref{table:comparison}) had too large witness spaces to conduct
these experiments within available resources.  For those benchmarks
where we have variance data, we observe that the variance obtained
using {\XORSampleprime} is larger (by upto a factor of \added{$43$}) than
those obtained using {\UniformWitness} in almost all cases.  Overall,
our experiments indicate that {\UniformWitness} always works
significantly faster and gives more (or comparably) uniformly
distributed witnesses vis-a-vis {\XORSampleprime} in almost all cases.
We also measured the probability of success of {\UniformWitness} for
each benchmark as the ratio of the number of runs for which the
algorithm did not return $\bot$ to the total number of runs.  We found
that \added{this exceeded} $0.6$ for every benchmark using
{\UniformWitness}.
%\vspace*{-0.15in}

\begin{figure}[h]
\begin{minipage}[b]{0.5\linewidth}
\centering
\includegraphics[width=\textwidth]{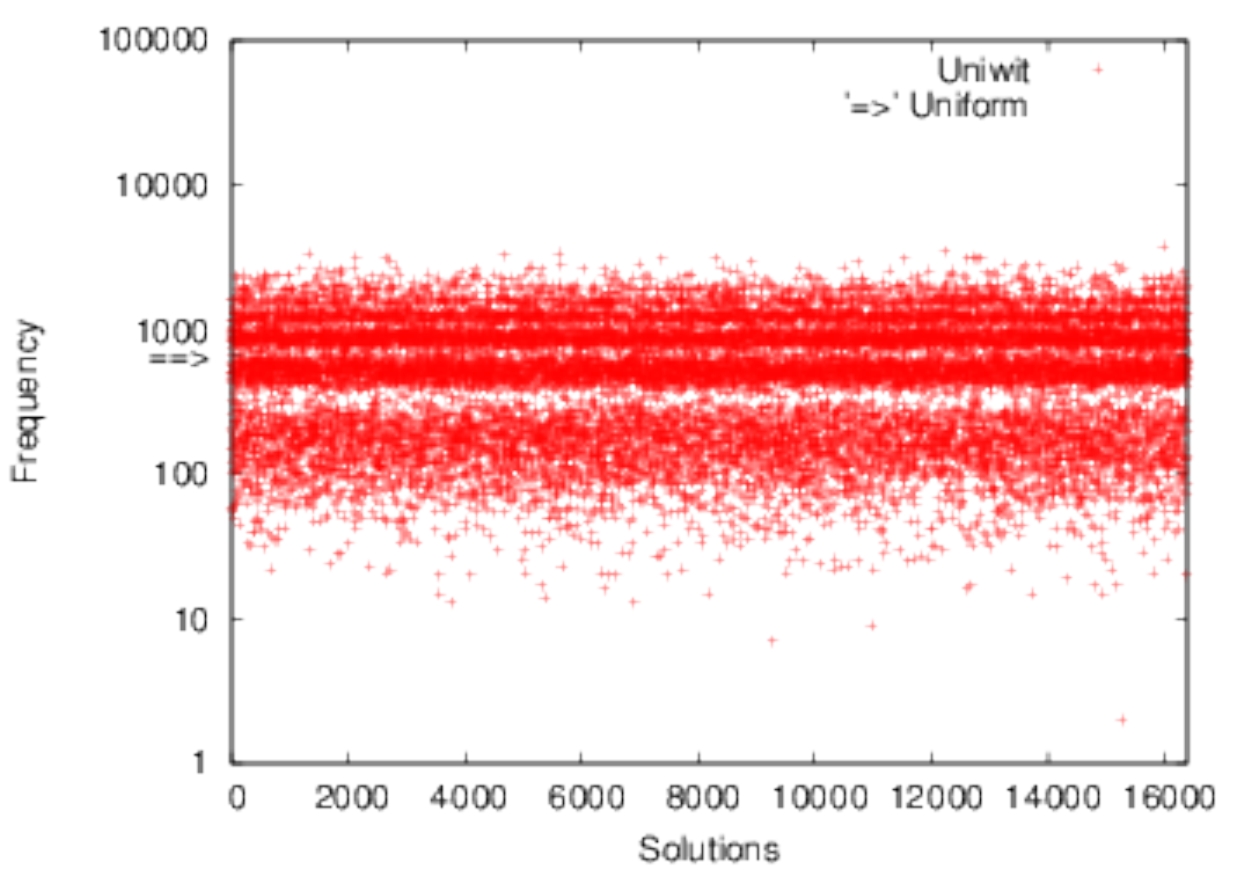}
\caption{Sampling by  {\UniformWitness} (k=2) }
\label{fig:Uniform}
\end{minipage}
\begin{minipage}[b]{0.5\linewidth}
\centering
\includegraphics[width=\textwidth]{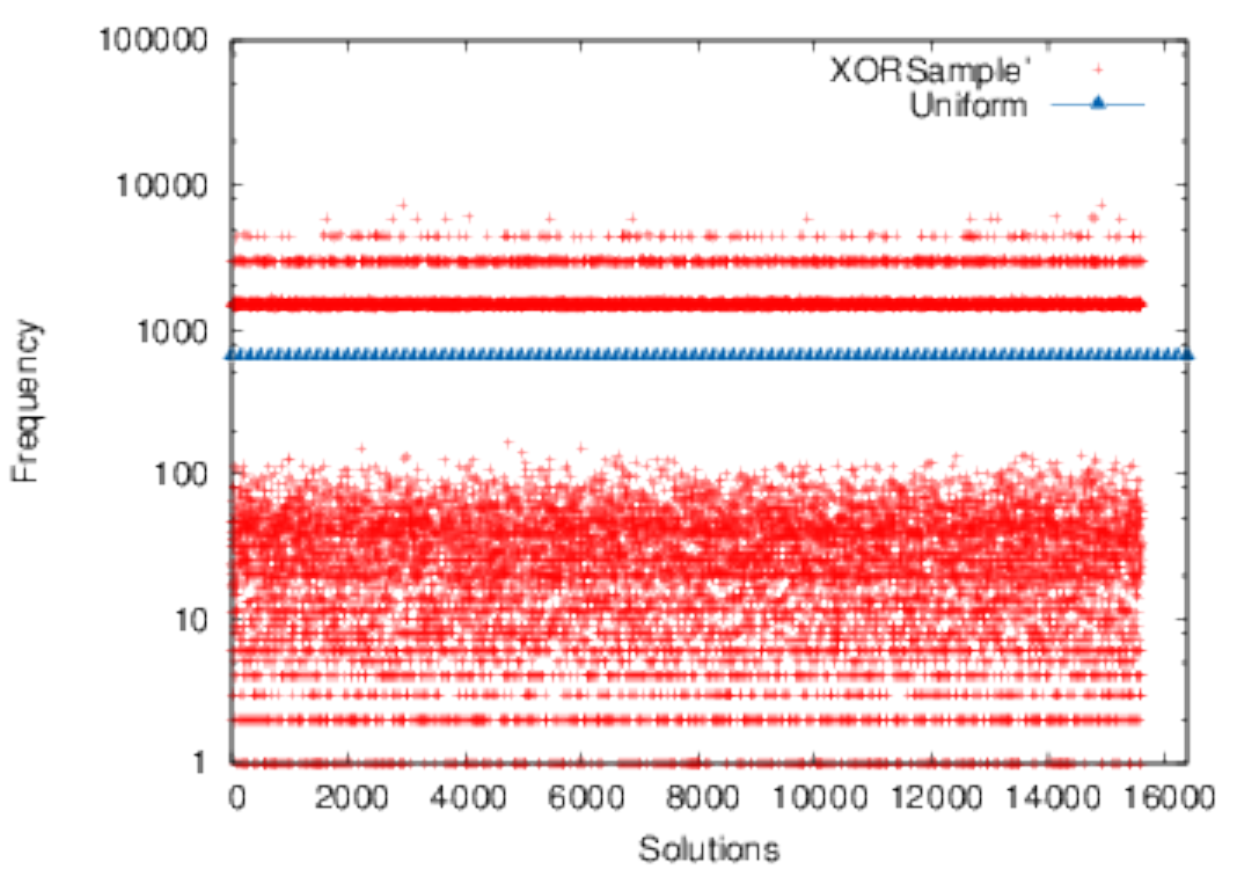}
\caption{Sampling by {\XORSampleprime}}
\label{fig:Selman}
\end{minipage}
\end{figure}
As an illustration of the difference in uniformity of witnesses
generated by {\UniformWitness} and {\XORSampleprime},
Figures~\ref{fig:Uniform} and \ref{fig:Selman} depict the frequencies
of appearance of various witnesses using these two algorithms for an
input CNF formula (case110) with $287$ variables and $16,384$
satisfying assignments.  The horizontal axis in each figure represents
witnesses numbered suitably, while the vertical axis
represents the generated frequencies of witnesses.  The frequencies
were obtained from $10.8 \times 10^6$ successful runs of each
algorithm.  Interestingly, {\XORSampleprime} could find only $15,612$
solutions (note the empty vertical band at the right end of
Figure~\ref{fig:Selman}), while {\UniformWitness} found all $16,384$
solutions.  Further, {\XORSampleprime} generated each of $15$
solutions more than $5,500$ times, and more than $250$ solutions were
generated only once.  No such major deviations from uniformity were
however observed in the frequencies generated by {\UniformWitness}.
\added{We also found that $15624$ out of $16384$ (i.e. $95.36\%$)
  witnesses generated by {\UniformWitness} had frequencies in excess
  of $N_{unif}/8$, where $N_{unif} = 10.8 \times 10^6/16384 \approx
  659$.}  In contrast, only $6047$ (i.e. $36.91\%$) witnesses
generated by {\XORSampleprime} had frequencies in excess of
$N_{unif}/8$.
\vspace*{-0.15in}
\section{Concluding Remarks}\label{sec:conclusion}

We described {\UniformWitness}, an algorithm that near-uniformly
samples random witnesses of Boolean formulas.  We showed that the
algorithm scales to reasonably large problems.  We also showed that it
performs better, in terms of both run time and uniformity, than
previous best-of-breed algorithms for this problem. The theoretical
guarantees can be further improved with higher independence of the
family of hash functions used in {\UniformWitness} (see \url{
  http://www.cfdvs.iitb.ac.in/reports/reports/CAV13/} for details).

\added{We have yet} to fully explore the parameter space and the effect of
pseudorandom generators other than HotBits for {\UniformWitness}.
There is a trade off between failure probability, time for first
witness, and time for subsequent witnesses.  During our experiments,
we observed the acute dearth of benchmarks available in the public
domain for this important problem. We hope that our work will lead to
development of benchmarks for this problem.  Our focus here has been
on Boolean constraints, which play a prominent role in hardware
design.  \added{Extending the algorithm to handle user-provided biases
  would be an interesting direction of future work.  Yet another
  interesting extension would be to consider richer constraint
  languages and build a uniform generator of witnesses modulo
  theories, leveraging recent progress in satisfiability modulo
  theories, c.f., \cite{dMB11}.}\\

\vspace*{-0.15in}

\bibliography{Report}

\begin{thebibliography}{10}

\bibitem{CryptoMiniSAT}
{CryptoMiniSAT}.
\newblock \url{http://www.msoos.org/cryptominisat2/}.

\bibitem{HotBits}
{HotBits}.
\newblock \url{http://www.fourmilab.ch/hotbits}.

\bibitem{SMTLib}
{SMTLib}.
\newblock \url{http://goedel.cs.uiowa.edu/smtlib/}.

\bibitem{Bacchus2003}
F.~Bacchus, S.~Dalmao, and T.~Pitassi.
\newblock Algorithms and complexity results for \#{SAT} and {B}ayesian
  inference.
\newblock In {\em Proc. of FOCS}, pages 340--351, 2003.

\bibitem{Bellare98uniformgeneration}
M.~Bellare, O.~Goldreich, and E.~Petrank.
\newblock Uniform generation of {NP}-witnesses using an {NP}-oracle.
\newblock {\em Information and Computation}, 163(2):510--526, 1998.

\bibitem{Ben05}
B.~Bentley.
\newblock Validating a modern microprocessor.
\newblock In {\em Proc. of CAV}, pages 2--4, 2005.

\bibitem{chandra-Verification}
A.K. Chandra and V.S. Iyengar.
\newblock Constraint solving for test case generation: A technique for
  high-level design verification.
\newblock In {\em Proc. of ICCD}, pages 245 --248, 1992.

\bibitem{chang2008functional}
K.~Chang, I.L. Markov, and V.~Bertacco.
\newblock {\em Functional Design Errors in Digital Circuits: Diagnosis
  Correction and Repair}.
\newblock Springer, 2008.

\bibitem{Darwiche02acompiler}
A.~Darwiche.
\newblock A compiler for deterministic, decomposable negation normal form.
\newblock In {\em Proc. of AAAI}, pages 627--634, 2002.

\bibitem{dMB11}
L.M. {de Moura} and N.~Bj{\o}rner.
\newblock Satisfiability {M}odulo {T}heories: {I}ntroduction and
  {A}pplications.
\newblock {\em Commun. ACM}, 54(9):69--77, 2011.

\bibitem{Gomes06modelcounting}
C.P. Gomes, A.~Sabharwal, and B.~Selman.
\newblock Model counting: A new strategy for obtaining good bounds.
\newblock In {\em Proc. of AAAI}, pages 54--61, 2006.

\bibitem{Gomes-Sampling}
C.P. Gomes, A.~Sabharwal, and B.~Selman.
\newblock Near-{U}niform sampling of combinatorial spaces using {XOR}
  constraints.
\newblock In {\em Proc. of NIPS.}, pages 670--676, 2007.

\bibitem{GABK11}
E.~Guralnik, M.~Aharoni, A.J. Birnbaum, and A.~Koyfman.
\newblock Simulation-based verification of floating-point division.
\newblock {\em IEEE Trans. on Computers}, 60(2):176--188, 2011.

\bibitem{Jerr}
M.R. Jerrum, L.G. Valiant, and V.V. Vazirani.
\newblock Random generation of combinatorial structures from a uniform
  distribution.
\newblock {\em Theoretical Computer Science}, 43(2-3):169--188, 1986.

\bibitem{Kirkpatrick83}
S.~Kirkpatrick, C.~D. Gelatt, and M.~P. Vecchi.
\newblock Optimization by simulated annealing.
\newblock {\em Science}, 220(4598):671--680, 1983.

\bibitem{KitKue2007}
N.~Kitchen and A.~Kuehlmann.
\newblock Stimulus generation for constrained random simulation.
\newblock In {\em Proc. of ICCAD}, pages 258--265, 2007.

\bibitem{Mansour}
Y.~Mansour, N.~Nisan, and P.~Tiwari.
\newblock The computational complexity of universal hashing.
\newblock {\em Theoretical Computer Science}, 107(1):235--243, 2002.

\bibitem{NRJKVMS06}
Y.~Naveh, M.~Rimon, I.~Jaeger, Y.~Katz, M.~Vinov, E.~Marcus, and G.~Shurek.
\newblock Constraint-based random stimuli generation for hardware verification.
\newblock In {\em Proc. of AAAI}, pages 1720--1727, 2006.

\bibitem{Roth1996}
D.~Roth.
\newblock On the hardness of approximate reasoning.
\newblock {\em Artificial Intelligence}, 82(1):273--302, April 1996.

\bibitem{Sipser83}
M.~Sipser.
\newblock A complexity theoretic approach to randomness.
\newblock In {\em Proc. of STOC}, pages 330--335, 1983.

\bibitem{Stockmeyer83}
L.~Stockmeyer.
\newblock The complexity of approximate counting.
\newblock In {\em Proc. of STOC}, pages 118--126, 1983.

\bibitem{Selman-Sampling}
W.~Wei, J.~Erenrich, and B.~Selman.
\newblock Towards efficient sampling: Exploiting random walk strategies.
\newblock In {\em Proc. of AAAI}, pages 670--676, 2004.

\bibitem{Yuan2004}
J.~Yuan, A.~Aziz, C.~Pixley, and K.~Albin.
\newblock Simplifying boolean constraint solving for random simulation-vector
  generation.
\newblock {\em IEEE Trans. on CAD of Integrated Circuits and Systems},
  23(3):412--420, 2004.

\end{thebibliography}
\bibliographystyle{plain}
%\clearpage
%\appendix
%\input{appendix}	
\end{document}